\documentclass[a4paper,11pt]{article}

\usepackage{amsmath,amsthm,amssymb,amsfonts}
\usepackage{bm}
\usepackage{geometry}
\usepackage{microtype}
\usepackage{url}
\usepackage{caption}
\newtheorem{thm}{Theorem}[section]
\newtheorem{lem}[thm]{Lemma}
\newtheorem{rem}[thm]{Remark}

\newcommand{\E}{\mathbb{E}}

\newcommand{\Id}{\mathrm{Id}}
\newcommand{\Impart}{\mathrm{Im}}

\newcommand{\dist}{\mathrm{dist}}

\newcommand{\length}{\mathrm{length}}

\title{Strong-disorder density-of-states expansions for the Anderson model
on regular trees and high-girth graphs}
\author{%
Masahiro Kaminaga\thanks{Corresponding author.
Email: kaminaga@g.tohoku-gakuin.ac.jp}\\
{\small Department of \textit{Information Technology}, Faculty of Engineering,}\\
{\small Tohoku Gakuin University, Sendai, Miyagi 984-8588, Japan}\\
{\small ORCID:\ 0000-0001-7204-8300}%
}
\date{}

\begin{document}

\maketitle

\begin{abstract}
We study strong-disorder expansions of the averaged local density of states for the Anderson model on the regular tree.
We assume that the single-site distribution has a density which is analytic on a prescribed energy interval, without requiring global analyticity.
A closed-walk expansion and the contour-deformation method for single-site Stieltjes transforms give a convergent series in the inverse disorder strength.
The coefficients are expressed in terms of occupation profiles of rooted closed walks, and all odd coefficients vanish on the tree.
We also compare the local densities on a regular graph and on its universal covering tree.
If the graph has girth $g$, their coefficients agree below order $g$, and the possible first difference is given explicitly by the cycles of length $g$ through the root.
We obtain a uniform error bound which is exponentially small in $g$ at sufficiently strong disorder.
The convergence estimate on the tree is sharpened using its adjacency-operator norm.
For a uniform single-site law, we calculate the density coefficients through fourth order and identify the contribution of four-cycles.


\end{abstract}

\section{Introduction}

Let $\mathbb T_q$ denote the infinite $(q+1)$-regular tree, where $q\geq2$.
We call it the Bethe lattice, using the terminology associated with Bethe's lattice models \cite{Bethe1935}, and fix a root $0$.
Let $A$ be the adjacency operator on $\ell^2(\mathbb T_q)$ and consider
\begin{equation}
\label{eq:intro-H}
H_{\lambda,\omega}=A+\lambda V_\omega,
\qquad
(V_\omega\psi)(x)=\omega_x\psi(x).
\end{equation}
Here the variables $\{\omega_x\}$ are independent and identically distributed with common probability law $\mu$, and $\lambda>0$ is the disorder strength.
We study the measure
\begin{equation}
\label{eq:intro-N}
N_\lambda(B)=
\mathbb E\langle\delta_0,{\bf1}_B(H_{\lambda,\omega})\delta_0\rangle
\end{equation}
for Borel sets $B\subset\mathbb R$.
This is the disorder-averaged spectral measure at the root; for general background on the spectral measure, see \cite{CyconFroeseKirschSimon1987}.
If its restriction to an open interval is absolutely continuous, we write $dN_\lambda(E)=n_\lambda(E)\,dE$ on that interval, and call $n_\lambda$ the root-averaged density of states.
This is the meaning of the term ``density of states'' throughout this paper.

The distinction from a finite-volume density of states is important.
The boundary of a ball in $\mathbb T_q$ has the same order as its volume.
Consequently, normalized eigenvalue counts for finite balls need not approach the measure in \eqref{eq:intro-N}; their limits are connected with the canopy graph \cite{AizenmanWarzel2006Canopy}.
The root-averaged measure is the natural local quantity considered here.
The tree Anderson model has the almost-sure spectrum
\[
\Sigma_\lambda=[-2\sqrt q,2\sqrt q]
                 +\lambda\operatorname{supp}\mu,
\]
under the present assumptions; see, for example, \cite{Klein1994,AizenmanWarzel2013}.

There is a long history of regularity results for density-of-states measures.
For lattice models, these include the work of Constantinescu, Fr\"ohlich, and Spencer \cite{ConstantinescuFrohlichSpencer1984} and the strong-disorder results of Bovier, Campanino, Klein, and Perez \cite{BovierCampaninoKleinPerez1988}.
The regularity of related correlations was studied by Bellissard and Hislop \cite{BellissardHislop2007}.
For the Bethe lattice, Acosta and Klein studied analyticity for Cauchy-type single-site laws \cite{AcostaKlein1992}, and Dolai and Krishna proved high-disorder smoothness under regularity assumptions on the single-site law \cite{DolaiKrishna2023}.
The spectral theory of random operators on trees is also developed in \cite{Klein1994,Klein1998,AizenmanSimsWarzel2006,AizenmanWarzel2013}.

The analytic continuation method used here is not new.
Kaminaga, Krishna, and Nakamura \cite{KaminagaKrishnaNakamura2012} proved local analyticity of the density of states on $\mathbb Z^d$ using a resolvent walk expansion and contour deformation of the single-site Stieltjes transforms.
We use the same contour argument.
Its extension to a regular tree is not, by itself, a new analyticity method.
The point of the present paper is to study the coefficients of the convergent strong-disorder expansion and to relate them to the local geometry of the underlying graph.
In addition to explicit tree coefficients, we prove a girth-dependent comparison for the averaged local density and calculate the first possible contribution of a shortest cycle.
These coefficient and density comparisons are not the results established in \cite{KaminagaKrishnaNakamura2012}.

The use of tree-like neighborhoods is also familiar in the study of large random regular graphs.
Geisinger \cite{Geisinger2015} considered both adjacency matrices and random Schr\"odinger operators on finite random regular graphs.
The main limit there is the number of vertices tending to infinity, and the principal comparison concerns integrated density-of-states distributions and their convergence to the corresponding infinite-tree distribution.
In particular, Theorem~3.1 of \cite{Geisinger2015} gives local integrated-spectral-measure comparisons for each fixed finite regular graph and each fixed vertex, including the case of independent random diagonal potentials.
Thus neither a fixed graph nor a fixed root alone distinguishes the present work.
Bauerschmidt, Huang, and Yau \cite{BauerschmidtHuangYau2019} studied local Green functions and a local Kesten--McKay law for the adjacency matrices of random regular graphs, without an independent diagonal Anderson potential.

Our comparison concerns a different observable and expansion parameter.
We fix an arbitrary connected $D$-regular graph, finite or countably infinite, without choosing it at random, and put independent random potentials on its vertices.
We let the disorder strength tend to infinity, rather than the graph size.
The object being compared is the disorder-averaged density of the spectral measure at a specified root, evaluated pointwise in a fixed scaled-energy interval.
We prove equality of the strong-disorder coefficients below the girth, give the first possible coefficient difference explicitly, and obtain an energy-uniform density estimate of order $(DK_\delta/\lambda)^g$.
The basic agreement of short walks with walks on the universal covering tree is a standard combinatorial fact; the new result here is its use in the locally analytic density expansion with an explicit first correction and a uniform remainder.
We do not study a graph-size limit or eigenvector delocalization.

Related expansion methods also address other questions and settings.
Hasler and Koberstein \cite{HaslerKoberstein2024} studied resolvent and integrated-density expansions for a continuum Schr\"odinger operator with a Poisson-distributed potential.
Fresta \cite{Fresta2021} used a supersymmetric cluster expansion to study averaged Green functions and regularity of local densities.
These results provide related context but do not yield the root-specific girth coefficient formula established below.

It is convenient to put $\varepsilon=\lambda^{-1}$ and to divide \eqref{eq:intro-H} by $\lambda$:
\begin{equation}
\label{eq:rescaled-H-intro}
\widetilde H_{\varepsilon,\omega}
=V_\omega+\varepsilon A
=\lambda^{-1}H_{\lambda,\omega}.
\end{equation}
Thus, the scaled energy $E=\lambda\xi$ corresponds to the fixed energy $\xi$ for $\widetilde H_{\varepsilon,\omega}$.
At $\varepsilon=0$, its root-averaged measure is exactly $\mu$.
The leading density coefficient is therefore $\rho$, whenever $\mu$ has a density $\rho$.
This leading-order fact is elementary; the higher corrections are our concern.

We first prove a convergent expansion of the averaged diagonal resolvent on a complex neighborhood of an interval where $\rho$ is locally analytic.
Its boundary values give a convergent expansion of the root-averaged density.
The coefficients are finite sums indexed by the occupation numbers of closed walks.
Bipartiteness implies that all odd coefficients vanish on the regular tree.
We sharpen the convergence estimate using the exact norm $2\sqrt q$ of its adjacency operator, instead of the degree bound $q+1$.

The main additional result concerns a connected $D$-regular graph of girth $g$.
The strong-disorder coefficients at a fixed root agree with those on the infinite $D$-regular tree through order $g-1$.
At order $g$, a term determined by the shortest cycles appears.
A uniform remainder bound quantifies the accuracy of the tree approximation in terms of $g$.
The fourth-order formula for a uniform single-site law gives a concrete example of this effect.

Section~\ref{sec:randomwalkexpansion} gives the walk expansion and its first coefficients.
We prove analytic continuation and the convergence estimate in the next two sections.
Section~\ref{sec:dos} establishes the density expansion, Section~\ref{sec:girth} proves the girth comparison, and Section~\ref{sec:uniform-distribution} calculates the uniform-law coefficients through fourth order.

%%%%%%%%%%%%%%%%%%%%%%%%%%%%%%%%%%%%%%%%%%%%%%%%%%%
\section{Model, notation, and main results}

Let $q\ge2$, let $\mathbb T_q$ be the infinite $(q+1)$-regular tree, and let $0\in\mathbb T_q$ be a fixed root.
We write $x\sim y$ if $x$ and $y$ are neighbors.
We consider the adjacency operator $A$ on $\ell^2(\mathbb T_q)$, defined by
$$
(A\psi)(x)=\sum_{y\sim x}\psi(y).
$$
The exact operator norm is
\begin{equation}
\label{eq:adjacency-tree-norm}
\|A\|=2\sqrt q;
\end{equation}
see, for example, \cite{Woess2000}.
We use this identity below to sharpen the bounds for closed walks.
Since
\[
(q+1)-2\sqrt q=(\sqrt q-1)^2>0\qquad(q\ge2),
\]
the bound from the operator norm is strictly smaller than the bound from the vertex degree for every $q$ considered here.
Let $\{\omega_x\}_{x\in\mathbb T_q}$ be independent identically distributed real random variables with common law $\mu$, and assume that $\mu$ has compact support.
For $\lambda>0$, we define the random Schr\"odinger operator
$$
H_{\lambda,\omega}=A+\lambda V_\omega.
$$
For $z\in\mathbb C_+=\{z\in\mathbb C:\Impart z>0\}$, we write
$$
G_{\lambda,\omega}(0,0;z)
=
\langle \delta_0,(H_{\lambda,\omega}-z)^{-1}\delta_0\rangle
$$
and
$$
m_\lambda(z)=\E G_{\lambda,\omega}(0,0;z).
$$
This averaged diagonal resolvent is the main object of this paper.

For $k=1,2,\dots$, we define the single-site Stieltjes transforms by
\begin{equation}
\label{eq:single-site-transform}
s_k(\zeta)
=
\int_{\mathbb R}\frac{d\mu(t)}{(t-\zeta)^k},
\qquad
\Impart \zeta>0.
\end{equation}
We now state the local analyticity assumption which is used in the rest of the paper.
Let $I=(a,b)$ be a bounded open interval, and let $\delta_0>0$.
We put
$$
I^{\sharp}=(a-\delta_0,b+\delta_0)
$$
and
$$
\Omega_\delta(I)=\{z\in\mathbb C:\dist(z,I)<\delta\}
$$
for $0<\delta\le \delta_0$.

We assume that $\mu$ has a density $\rho$ on $I^{\sharp}$ and that this density, still denoted by $\rho$, extends holomorphically to $\Omega_{\delta_0}(I)$ and continuously to its closure.
Outside $I^{\sharp}$, the measure $\mu$ may be singular.
This local assumption is sufficient for our argument.

For later use, we introduce the notation for closed walks.
For $n\ge0$, let $\Gamma_n(0,0)$ be the set of all closed walks of length $n$ starting and ending at the root $0$, that is,
$$
\Gamma_n(0,0)=\{\gamma=(x_0,\ldots,x_n): x_0=x_n=0,\ 
x_{j-1}\sim x_j,\ j=1,\ldots,n\}.
$$
For $\gamma\in\Gamma_n(0,0)$, we define the occupation number of a vertex $x$ by
$$
\nu_\gamma(x) = \#\{j\in\{0,\ldots,n\}: x_j=x\}.
$$
The occupation profile of $\gamma$ means the finite list of the positive occupation numbers $\nu_\gamma(x)$, written in nonincreasing order.
Equivalently, it is encoded by $m_k(\gamma)=\#\{x\in \mathbb T_q:\nu_\gamma(x)=k\}$.

The first theorem is the main resolvent statement.

\begin{thm}\label{thm:main-resolvent}
Assume the local analyticity condition above.
Fix $0<\delta<\delta_0$.
Then, there exist $\lambda_0>0$ and functions $M_n$, $n=0,1,2,\ldots$, holomorphic on $\Omega_\delta(I)$ such that, for every $\lambda\ge\lambda_0$, the function
$$
\zeta\mapsto m_\lambda(\lambda\zeta),
\qquad \Impart\zeta>0,
$$
extends holomorphically to $\mathbb C_+\cup\Omega_\delta(I)$, and the series
$$
\sum_{n=0}^{\infty}\lambda^{-n-1}M_n(\zeta)
$$
converges uniformly on $\Omega_\delta(I)$ to this holomorphic continuation.

Moreover, for every $N\ge0$, there exists $C_{N,\delta}>0$ such that
\begin{equation}
\label{eq:main-resolvent-expansion-zeta}
m_\lambda(\lambda\zeta)=\sum_{n=0}^{N}\lambda^{-n-1}M_n(\zeta)
+\widetilde R_{N,\lambda}(\zeta),\qquad\zeta\in\Omega_\delta(I),
\end{equation}
with the uniform bound
$$
|\widetilde R_{N,\lambda}(\zeta)|
\le
C_{N,\delta}\lambda^{-N-2}.
$$
The coefficient functions are given by
\begin{equation}
\label{eq:Mn-formula-statement}
M_n(\zeta)
=
(-1)^n
\sum_{\gamma\in\Gamma_n(0,0)}
\prod_{x:\,\nu_\gamma(x)>0}
s_{\nu_\gamma(x)}(\zeta),
\qquad
\zeta\in\Omega_\delta(I).
\end{equation}
In particular, $M_n\equiv0$ for every odd $n$, and
$$
M_0(\zeta)=s_1(\zeta).
$$
\end{thm}

It is useful to rewrite \eqref{eq:main-resolvent-expansion-zeta} in the original variable $z$.
Replacing $\zeta$ by $z/\lambda$ in \eqref{eq:main-resolvent-expansion-zeta}, we obtain
\begin{equation}
\label{eq:main-resolvent-expansion}
m_\lambda(z)=\sum_{n=0}^{N}\lambda^{-n-1}M_n(z/\lambda)+R_{N,\lambda}(z),
\qquad z/\lambda\in\Omega_\delta(I),
\end{equation}
where
$$
R_{N,\lambda}(z)=\widetilde R_{N,\lambda}(z/\lambda),
$$
and
\begin{equation}
\label{eq:main-resolvent-remainder}
|R_{N,\lambda}(z)|
\le
C_{N,\delta}\lambda^{-N-2}.
\end{equation}

Before stating the consequence for the density of states, we give a formal definition.
Let $E_{H_{\lambda,\omega}}(B)$ denote the spectral projection of $H_{\lambda,\omega}$ for a Borel set $B$ and put
\begin{equation}
\label{eq:N-formal-def}
N_\lambda(B)=
\E\langle\delta_0,E_{H_{\lambda,\omega}}(B)\delta_0\rangle.
\end{equation}
If $N_\lambda$ is absolutely continuous on an open set $J$, its Radon--Nikodym density $n_\lambda$ is characterized by
\[
N_\lambda(B)=\int_B n_\lambda(E)\,dE,
\qquad B\subset J \text{ Borel}.
\]
We use the real-analytic representative of this density on $J=\lambda I$ below.
This definition does not refer to a finite-volume eigenvalue counting limit.

The next theorem is the consequence for the density of states.

\begin{thm}
\label{thm:main-dos}
Assume the same hypotheses as in Theorem~\ref{thm:main-resolvent}.
Fix $0<\delta<\delta_0$ and $N\ge0$, and let $\lambda_0$ and $C_{N,\delta}$ be the constants in Theorem~\ref{thm:main-resolvent}.
Then, for every $\lambda\ge\lambda_0$, the measure $N_\lambda$ is absolutely continuous on $\lambda I$, and its density $n_\lambda(E)$ is real analytic there.
More precisely, for $\lambda\ge\lambda_0$ and $\xi\in I$, we have the convergent series
\begin{equation}
\label{eq:main-dos-convergent}
n_\lambda(\lambda\xi)
=\sum_{n=0}^{\infty}\lambda^{-n-1}a_n(\xi),
\end{equation}
where the convergence is uniform in $\xi\in I$.
In particular, for any $N\geq0$,
\begin{equation}
\label{eq:main-dos-expansion}
n_\lambda(\lambda\xi)
=
\sum_{n=0}^{N}\lambda^{-n-1}a_n(\xi)
+
r_{N,\lambda}(\xi),
\end{equation}
where
\begin{equation}
\label{eq:an-def}
a_n(\xi)
=
\frac{1}{\pi}\Impart M_n(\xi),
\qquad
\xi\in I.
\end{equation}
Here $M_n(\xi)$ denotes the value at $\xi\in I$ of the holomorphic continuation of $M_n$ obtained in Theorem~\ref{thm:main-resolvent}.
The remainder satisfies
\begin{equation}
\label{eq:rn-bound}
|r_{N,\lambda}(\xi)|
\le
\frac{1}{\pi}C_{N,\delta}\lambda^{-N-2}.
\end{equation}
In particular,
\begin{equation}
\label{eq:a0-rho}
a_0(\xi)=\rho(\xi),
\end{equation}
and $a_n\equiv0$ for every odd $n$.
\end{thm}

We also state a comparison with other regular graphs.
It is a coefficient-level statement, so it does not require a limiting sequence of graphs.
For a rooted regular graph $(G,o)$ we use $N_{\lambda}^{(G,o)}$ for the averaged local spectral measure and $n_{\lambda}^{(G,o)}$ for its density.
We put
\[
f_{\lambda}^{(G,o)}(\xi)=
\lambda n_{\lambda}^{(G,o)}(\lambda\xi)
\]
and write $f_{\lambda}^{(\mathbb T_{D-1})}$ for the corresponding scaled density on the infinite $D$-regular tree.

\begin{thm}\label{thm:girth-comparison}
Let $D\ge3$ and let $G$ be a connected, simple, $D$-regular graph, finite or countably infinite, with finite girth $g\ge3$.
Fix a root $o$.
On $G$ and on $\mathbb T_{D-1}$, consider independent single-site potentials with the same law $\mu$ satisfying the local analyticity assumption above.
Let $c_g(o)$ be the number of oriented simple cycles of length $g$ starting and ending at $o$.
Then the local resolvent and density coefficients on the two graphs satisfy
\begin{equation}\label{eq:girth-coeff-match}
M_n^{(G,o)}(\zeta)=M_n^{(\mathbb T_{D-1})}(\zeta),
\qquad 0\le n<g,
\end{equation}
while
\begin{equation}\label{eq:girth-first-difference}
M_g^{(G,o)}(\zeta)-M_g^{(\mathbb T_{D-1})}(\zeta)
=(-1)^g c_g(o)s_2(\zeta)s_1(\zeta)^{g-1}.
\end{equation}
These identities hold on $\Omega_\delta(I)$ for every $0<\delta<\delta_0$.

More precisely, let $K_\delta$ be the constant in Lemma~\ref{lem:Mn-bound}, and put
\[
\lambda_*=2D\max\{K_\delta,\delta^{-1}\}.
\]
For every $\lambda\ge\lambda_*$, both scaled densities are real analytic on $I$, and, uniformly for $\xi\in I$,
\begin{equation}\label{eq:girth-uniform-bound}
\left|f_{\lambda}^{(G,o)}(\xi)-
 f_{\lambda}^{(\mathbb T_{D-1})}(\xi)\right|
\le \frac{4K_\delta}{\pi}
\left(\frac{DK_\delta}{\lambda}\right)^g.
\end{equation}
The possible first nonzero correction is explicit:
\begin{equation}\label{eq:girth-leading-dos}
\begin{aligned}
f_{\lambda}^{(G,o)}(\xi)-
 f_{\lambda}^{(\mathbb T_{D-1})}(\xi)
={}&\frac{(-1)^g c_g(o)}{\pi\lambda^g}
\Impart\!\left[s_2(\xi)s_1(\xi)^{g-1}\right]
+\mathcal R_{g,\lambda}(\xi),
\end{aligned}
\end{equation}
where
\begin{equation}\label{eq:girth-next-remainder}
|\mathcal R_{g,\lambda}(\xi)|
\le\frac{4K_\delta}{\pi}
\left(\frac{DK_\delta}{\lambda}\right)^{g+1}.
\end{equation}
All constants in these bounds are independent of $G$, $o$, and $g$.
\end{thm}

\begin{rem}
Let us clarify the meaning of the density of states measure used in this paper.
On amenable graphs such as $\mathbb Z^d$, the integrated density of states is usually obtained as a normalized eigenvalue counting limit over finite boxes, because the boundary-volume ratio tends to zero.
The Bethe lattice is different.
Let $d$ denote the graph distance on $\mathbb T_q$.
For $r\ge0$, let
$$
S_r=\{x\in\mathbb T_q:d(x,0)=r\}
$$
be the sphere of radius $r$ around the root.
Since the root has $q+1$ neighbors and, after the first step, each vertex has one edge back to its parent and $q$ forward edges, one has
$$
|S_R|=(q+1)q^{R-1}\qquad (R\ge 1).
$$
If $\Lambda_R$ is the ball of radius $R$ around the root, then $\Lambda_R$ is the disjoint union of $S_0,S_1,\ldots,S_R$, and hence
$$
|\Lambda_R|=1+\sum_{r=1}^R (q+1)q^{r-1}=1+\frac{q+1}{q-1}(q^R-1).
$$
Therefore, we have
$$
\frac{|S_R|}{|\Lambda_R|}\to\frac{q-1}{q}
$$
as $R\to\infty$.
Thus, the boundary has a volume comparable with the volume of the ball.
Boundary effects do not disappear in normalized finite-volume eigenvalue counting measures, and such limits, if considered, may depend on the exhaustion and on the boundary condition.
The canopy-graph limit and its spectral properties are discussed in \cite{AizenmanWarzel2006Canopy}.

The measure $N_\lambda$ used here is instead a local object.
It is the disorder average of the spectral measure seen from the root.
Since the Bethe lattice is transitive and the random potential is identically distributed at all vertices, the root represents a typical vertex in this local sense.
Equivalently, $N_\lambda$ is the measure whose Stieltjes transform is the averaged diagonal Green function $m_\lambda(z)$.
Thus, the density $n_\lambda(E)$, when it exists, describes the averaged local density of states at a typical vertex, rather than a normalized eigenvalue counting measure over finite balls.
\end{rem}

\begin{rem}
The interval $\lambda I$ moves linearly with the disorder strength.
With $\varepsilon=\lambda^{-1}$, the rescaled operator is $\widetilde H_{\varepsilon,\omega}=V_\omega+\varepsilon A$.
If $\widetilde n_\varepsilon$ denotes its root-averaged density, then
\[
\widetilde n_\varepsilon(\xi)=\lambda n_\lambda(\lambda\xi)
=\sum_{n=0}^{\infty}\varepsilon^n a_n(\xi).
\]
In particular, the leading coefficient $a_0=\rho$ is the density at $\varepsilon=0$.
For a truncation after $a_{2m}$, we may apply the theorem with $N=2m+1$ because the odd coefficients vanish.
The corresponding remainder for $n_\lambda$ is $O(\lambda^{-2m-3})$.
\end{rem}

\section{Random walk expansion and coefficient functions}\label{sec:randomwalkexpansion}

We first recall the random-walk expansion in the region where the Neumann series converges in operator norm.
Put
$$
D_{\lambda,\omega}(z)=\lambda V_\omega-z.
$$
Since $\omega_x\in\mathbb R$ and $\Impart z>0$, the operator $D_{\lambda,\omega}(z)$ is invertible, and its inverse is the diagonal multiplication operator
$$
(D_{\lambda,\omega}(z)^{-1}\psi)(x)=\frac{1}{\lambda\omega_x-z}\psi(x).
$$
In particular,
$$
\|D_{\lambda,\omega}(z)^{-1}\|=\sup_{x\in\mathbb T_q}\frac{1}{|\lambda\omega_x-z|}
\le\frac{1}{\Impart z}.
$$
Therefore, we have
$$
H_{\lambda,\omega}-z=
\bigl(\Id+A D_{\lambda,\omega}(z)^{-1}\bigr)D_{\lambda,\omega}(z).
$$
If $\Impart z>2\sqrt q$, then
$$
\|A D_{\lambda,\omega}(z)^{-1}\|
\le
\|A\|\|D_{\lambda,\omega}(z)^{-1}\|
\le
\frac{2\sqrt q}{\Impart z}<1.
$$
Hence,
\begin{equation}
\label{eq:resolvent-neumann}
(H_{\lambda,\omega}-z)^{-1}=\sum_{n=0}^{\infty}(-1)^n
D_{\lambda,\omega}(z)^{-1} \bigl(A D_{\lambda,\omega}(z)^{-1}\bigr)^n.
\end{equation}

We next write the diagonal matrix element of the $n$th term explicitly.
For the canonical basis vectors, one has
$$
\langle \delta_x,A\delta_y\rangle =
\begin{cases}
1, & x\sim y,\\
0, & \text{otherwise},
\end{cases}
$$
and
$$
\langle \delta_x,D_{\lambda,\omega}(z)^{-1}\delta_y\rangle
=
\frac{\delta_{xy}}{\lambda\omega_x-z}.
$$
Hence, for $n\ge1$, after inserting the resolution of identity with respect to the canonical basis $\{\delta_x\}_{x\in\mathbb T_q}$ between successive factors, we obtain
$$
\langle \delta_0,
D_{\lambda,\omega}(z)^{-1}
\bigl(A D_{\lambda,\omega}(z)^{-1}\bigr)^n
\delta_0\rangle
=
\sum_{x_1,\dots,x_{n-1}\in\mathbb T_q}
\prod_{j=1}^{n}\langle \delta_{x_{j-1}},A\delta_{x_j}\rangle
\prod_{j=0}^{n}\frac{1}{\lambda\omega_{x_j}-z},
$$
where $x_0=x_n=0$.
The product of the matrix elements of $A$ is equal to $1$ exactly when
$$
x_{j-1}\sim x_j,
\qquad
j=1,\dots,n,
$$
and is equal to $0$ otherwise.
Therefore, the nonzero terms are in one to one correspondence with closed walks
$$
\gamma=(x_0,x_1,\dots,x_n)
$$
of length $n$ such that $x_0=x_n=0$ and $x_{j-1}\sim x_j$ for $j=1,\dots,n$.
These are precisely the walks in $\Gamma_n(0,0)$.
The case $n=0$ is simply
$$
\langle \delta_0,D_{\lambda,\omega}(z)^{-1}\delta_0\rangle
=
\frac{1}{\lambda\omega_0-z},
$$
and, for all $n\ge0$, we may write
$$
\langle \delta_0,
D_{\lambda,\omega}(z)^{-1}
\bigl(A D_{\lambda,\omega}(z)^{-1}\bigr)^n
\delta_0\rangle
=
\sum_{\gamma\in\Gamma_n(0,0)}
\prod_{j=0}^{n}\frac{1}{\lambda\omega_{x_j}-z}.
$$
Substituting this into \eqref{eq:resolvent-neumann}, we obtain
\begin{equation}
\label{eq:path-expansion-z}
G_{\lambda,\omega}(0,0;z)
=
\sum_{n=0}^{\infty}(-1)^n
\sum_{\gamma\in\Gamma_n(0,0)}
\prod_{j=0}^{n}\frac{1}{\lambda\omega_{x_j}-z}.
\end{equation}
The series is absolutely convergent.
Indeed,
$$
\left|
\prod_{j=0}^{n}\frac{1}{\lambda\omega_{x_j}-z}
\right|
\le
(\Impart z)^{-n-1}
$$
for every closed walk $\gamma$, and
\begin{equation}\label{eq:closed-walk-spectral-radius}
|\Gamma_n(0,0)|
=\langle\delta_0,A^n\delta_0\rangle
\le\|A\|^n=(2\sqrt q)^n.
\end{equation}
Hence
$$
\sum_{n=0}^{\infty}
\sum_{\gamma\in\Gamma_n(0,0)}
\left|
\prod_{j=0}^{n}\frac{1}{\lambda\omega_{x_j}-z}
\right|
\le
\sum_{n=0}^{\infty}
\frac{(2\sqrt q)^n}{(\Impart z)^{n+1}},
$$
and the right hand side converges when $\Impart z>2\sqrt q$.
This bound is independent of $\omega$.
Therefore, expectation may be taken term by term.

Recall that $\nu_\gamma(x)$ is the number of visits of $\gamma$ to $x$.
Then,
$$
\sum_{x\in\mathbb T_q}\nu_\gamma(x)=n+1.
$$
Since the random variables $\{\omega_x\}$ are independent, for each $\gamma\in\Gamma_n(0,0)$ we have
\begin{equation}
\label{eq:averaged-walk-weight}
\E\prod_{j=0}^{n}\frac{1}{\lambda\omega_{x_j}-z}
=
\prod_{x:\,\nu_\gamma(x)>0}
\int_{\mathbb R}
\frac{d\mu(t)}{(\lambda t-z)^{\nu_\gamma(x)}}.
\end{equation}
After the scaling $z=\lambda\zeta$, one has
$$
\frac{1}{\lambda t-\lambda\zeta}
=
\lambda^{-1}\frac{1}{t-\zeta}.
$$
Hence, for $\Impart \zeta>2\sqrt q/\lambda$,
\begin{equation}
\label{eq:scaled-series}
m_\lambda(\lambda\zeta)
=
\sum_{n=0}^{\infty}\lambda^{-n-1}M_n(\zeta),
\end{equation}
where
\begin{equation}
\label{eq:Mn-def}
M_n(\zeta)
=
(-1)^n
\sum_{\gamma\in\Gamma_n(0,0)}
\prod_{x:\,\nu_\gamma(x)>0}
s_{\nu_\gamma(x)}(\zeta),
\qquad
\Impart \zeta>0.
\end{equation}
Since $\Gamma_n(0,0)$ is finite for each fixed $n$, $M_n$ is holomorphic on $\mathbb C_+$.

The tree structure gives an immediate consequence.

\begin{lem}
\label{lem:odd-vanish}
For every odd integer $n$, one has $\Gamma_n(0,0)=\emptyset$.
Hence, $M_n\equiv0$ for every odd $n$.
\end{lem}

\begin{proof}
Let $|x|$ denote the graph distance from $0$ to $x$.
Along any edge move, the parity of $|x|$ changes.
Therefore, a walk of odd length starting at $0$ ends at a vertex whose distance from $0$ is odd, and it cannot return to $0$.
This proves $\Gamma_n(0,0)=\emptyset$ for odd $n$.
The formula \eqref{eq:Mn-def} then gives $M_n\equiv0$.
\end{proof}

The first coefficients are easily computed.

\begin{lem}
\label{lem:low-coefficients}
One has
\begin{equation}
\label{eq:M0-formula}
M_0(\zeta)=s_1(\zeta),
\end{equation}
\begin{equation}
\label{eq:M2-formula}
M_2(\zeta)=(q+1)s_2(\zeta)s_1(\zeta),
\end{equation}
and
\begin{eqnarray}
\label{eq:M4-formula}
M_4(\zeta)
&=&
(q+1)s_3(\zeta)s_2(\zeta)
+
(q+1)q\,s_3(\zeta)s_1(\zeta)^2
\nonumber\\
&&
+
(q+1)q\,s_2(\zeta)^2s_1(\zeta).
\end{eqnarray}
\end{lem}

\begin{proof}
The formula for $M_0$ follows from the trivial walk of length $0$.
There is no walk of length $1$ by Lemma~\ref{lem:odd-vanish}.
A walk of length $2$ has the form $0\to x\to0$ with $x\sim0$, and there are exactly $q+1$ such walks.
The root is visited twice and the vertex $x$ is visited once, so \eqref{eq:M2-formula} follows.

For walks of length $4$, there are three types.
The first type is $0\to x\to0\to x\to0$, the second type is $0\to x\to0\to y\to0$ with $x\ne y$, and the third type is $0\to x\to u\to x\to0$ with $u\sim x$ and $u\ne0$.
The numbers of such walks are $q+1$, $(q+1)q$, and $(q+1)q$, respectively.
Their occupation number profiles are $(3,2)$, $(3,1,1)$, and $(2,2,1)$.
Substituting these profiles into \eqref{eq:Mn-def} gives \eqref{eq:M4-formula}.
\end{proof}

\begin{rem}\label{rem:four-cycles}
The fourth-order formula gives the first example of the girth comparison in Theorem~\ref{thm:girth-comparison}.
For a simple $D$-regular graph with a reference vertex $o$, let $c_4(o)$ denote the number of oriented closed walks $(o,x,u,y,o)$ with four distinct vertices.
Such a walk goes around a four-cycle.
The same walk classification gives
\begin{equation}
\label{eq:M4-general-graph}
\begin{aligned}
M_4^{(G)}(\zeta)={}&D s_3s_2+D(D-1)s_3s_1^2\\
&+D(D-1)s_2^2s_1+c_4(o)s_2s_1^3,
\end{aligned}
\end{equation}
where all $s_k$ are evaluated at $\zeta$.
On $\mathbb T_q$, $D=q+1$ and $c_4(o)=0$, which gives \eqref{eq:M4-formula}.
In contrast, on $\mathbb Z^d$, $D=2d$ and $c_4(o)=4d(d-1)$.
The last term in \eqref{eq:M4-general-graph} is therefore absent on a tree.
This is a distinction in local closed-walk geometry.
It does not, by itself, characterize nonamenability; the vanishing of odd terms only uses bipartiteness.
\end{rem}

The general coefficient has the same form.

\begin{lem}
\label{lem:coefficient-structure}
For every even integer $n\ge0$, the function $M_n(\zeta)$ is a finite sum of monomials in $s_1(\zeta),\dots,s_{n+1}(\zeta)$, and the coefficient of each monomial is a nonnegative integer determined only by the number of closed walks with the corresponding occupation profile.
\end{lem}

\begin{proof}
Fix an even integer $n\ge0$.
For a closed walk $\gamma\in\Gamma_n(0,0)$, set
$$
m_k(\gamma)
=
\#\{x\in\mathbb T_q:\nu_\gamma(x)=k\},
\qquad
k=1,2,\dots,n+1.
$$
Then
$$
\sum_{k=1}^{n+1}k\,m_k(\gamma)=n+1.
$$
The contribution of $\gamma$ to $M_n(\zeta)$ is
$$
\prod_{k=1}^{n+1}s_k(\zeta)^{m_k(\gamma)}.
$$
Since $\Gamma_n(0,0)$ is finite, grouping the walks with the same vector $(m_1(\gamma),\dots,m_{n+1}(\gamma))$ proves the assertion.
\end{proof}

\section{Holomorphic continuation of the single-site transforms}

The key point is that the local analyticity of $\rho$ enables us to continue each $s_k$ across the interval $I$.
This contour-deformation argument is taken from \cite{KaminagaKrishnaNakamura2012}.

Fix $0<\delta<\delta_0$.
This inequality is needed in the bounds below; in particular, the constants depend on $\delta_0-\delta>0$.
Recall that
$$
I^\sharp=(a-\delta_0,b+\delta_0).
$$
Let $\eta$ be the lower part of the boundary of $\Omega_{\delta_0}(I)$ which connects $a-\delta_0$ to $b+\delta_0$, and orient $\eta$ from $a-\delta_0$ to $b+\delta_0$.
Thus, we have
$$
\eta \subset \{z\in\mathbb C:\operatorname{dist}(z,I)=\delta_0,\ \Im z\le 0\}.
$$
We also orient the interval $I^\sharp=(a-\delta_0,b+\delta_0)$ on the real axis from $a-\delta_0$ to $b+\delta_0$.
With these conventions, the closed contour $I^\sharp\cup(-\eta)$ is the negatively oriented boundary of the lower part of $\Omega_{\delta_0}(I)$; see Fig.~\ref{fig:contour-single-site}.
Hence Cauchy's theorem will give the same sign as in \eqref{eq:sk-continuation} below.

% This diagram uses only the standard LaTeX picture environment.
% Its bounding box is fixed, so no external graphic or TikZ is needed.
\begin{figure}[!htbp]
\centering
\begingroup
\setlength{\unitlength}{1mm}
\begin{picture}(125,40)
  % Labels for the nested intervals (all within the fixed picture box).
  \put(62.5,36){\makebox(0,0){$I^\sharp=(a-\delta_0,b+\delta_0)$}}
  \put(62.5,23){\makebox(0,0){$I=(a,b)$}}
  % The real axis and the endpoints of the larger interval.
  \put(10,18){\vector(1,0){105}}
  \put(119,18){\makebox(0,0)[l]{$\mathbb R$}}
  \put(31,16){\line(0,1){4}}
  \put(94,16){\line(0,1){4}}
  % The smaller interval is shown as a thick segment.
  \linethickness{1.1pt}
  \put(43,18){\line(1,0){39}}
  \linethickness{0.5pt}
  \put(43,16){\makebox(0,0)[t]{$a$}}
  \put(82,16){\makebox(0,0)[t]{$b$}}
  % Equal margins delta_0 to the left and right of I.
  \put(37,29){\vector(-1,0){6}}
  \put(37,29){\vector(1,0){6}}
  \put(88,29){\vector(-1,0){6}}
  \put(88,29){\vector(1,0){6}}
  \put(37,32){\makebox(0,0){$\delta_0$}}
  \put(88,32){\makebox(0,0){$\delta_0$}}
  % The lower deformation contour, oriented from left to right.
  \qbezier(31,18)(31,6)(43,6)
  \put(43,6){\line(1,0){39}}
  \qbezier(82,6)(94,6)(94,18)
  \put(62.5,6){\vector(1,0){10}}
  \put(62.5,2.5){\makebox(0,0){$\eta$}}
\end{picture}
\endgroup
\caption{The interval $I$, the larger interval $I^\sharp$, and the
lower contour $\eta$ oriented from left to right.
The integral over
$I^\sharp$ is deformed to the integral over $\eta$}
\label{fig:contour-single-site}
\end{figure}

\begin{lem}
\label{lem:single-site-continuation}
Fix $0<\delta<\delta_0$.
For each $k=1,2,\dots$, the function $s_k$ defined on $\mathbb C_+$ by \eqref{eq:single-site-transform} extends holomorphically to $\mathbb C_+\cup\Omega_\delta(I)$.
More precisely, for $\zeta\in\Omega_\delta(I)$,
\begin{equation}
\label{eq:sk-continuation}
s_k(\zeta)
=
\int_{\mathbb R\setminus I^{\sharp}}
\frac{d\mu(t)}{(t-\zeta)^k}
+
\int_{\eta}\frac{\rho(w)\,dw}{(w-\zeta)^k}.
\end{equation}
Moreover, there exists a constant $C_\delta>0$ such that
\begin{equation}
\label{eq:sk-bound}
|s_k(\zeta)|
\le
C_\delta(\delta_0-\delta)^{-k},
\qquad
\zeta\in\Omega_\delta(I),
\quad
k=1,2,\dots.
\end{equation}
\end{lem}

\begin{proof}
For $\Impart \zeta>0$, we divide the integral defining $s_k(\zeta)$ into the part on $I^{\sharp}$ and the part on $\mathbb R\setminus I^{\sharp}$.
The contour $I^\sharp\cup(-\eta)$ is negatively oriented.
Since $\rho$ is holomorphic on $\Omega_{\delta_0}(I)$ and continuous on its closure, we may first use interior approximating contours and then pass to the boundary contour.
Cauchy's theorem therefore gives
$$
\int_{I^\sharp}\frac{\rho(t)\,dt}{(t-\zeta)^k}
-\int_{\eta}\frac{\rho(w)\,dw}{(w-\zeta)^k}=0.
$$
Thus, we have
$$
\int_{I^\sharp}\frac{\rho(t)\,dt}{(t-\zeta)^k}
=
\int_{\eta}\frac{\rho(w)\,dw}{(w-\zeta)^k}.
$$
Hence, for $\Impart \zeta>0$,
$$
s_k(\zeta)
=
\int_{\mathbb R\setminus I^{\sharp}}
\frac{d\mu(t)}{(t-\zeta)^k}
+
\int_{\eta}\frac{\rho(w)\,dw}{(w-\zeta)^k}.
$$
The right hand side is holomorphic for $\zeta\in\Omega_\delta(I)$.
Indeed, the set $(\mathbb R\setminus I^\sharp)\cup\eta$ is contained in the complement of $\Omega_{\delta_0}(I)$, and hence every $\zeta\in\Omega_\delta(I)$ has distance at least $\delta_0-\delta$ from this set.
Therefore, the two integrals define holomorphic functions of $\zeta$ on $\Omega_\delta(I)$.
Thus, the right hand side gives the desired continuation, and \eqref{eq:sk-continuation} follows.

We next prove the bound.
By the definition of $I^\sharp$ and $\eta$, both $\mathbb R\setminus I^\sharp$ and $\eta$ are contained in the complement of $\Omega_{\delta_0}(I)$.
Hence, if $\zeta\in\Omega_\delta(I)$, then
$$
|t-\zeta|\ge \delta_0-\delta,
\qquad
t\in\mathbb R\setminus I^\sharp,
$$
and
$$
|w-\zeta|\ge \delta_0-\delta,
\qquad
w\in\eta.
$$
Since $\mu$ is a probability measure, we obtain
\begin{eqnarray}
|s_k(\zeta)|
&\le&
\int_{\mathbb R\setminus I^{\sharp}}
\frac{d\mu(t)}{|t-\zeta|^k}
+
\int_{\eta}\frac{|\rho(w)|\,|dw|}{|w-\zeta|^k}
\nonumber\\
&\le&
(\delta_0-\delta)^{-k}
\left(
1+\length(\eta)\sup_{w\in\eta}|\rho(w)|
\right).
\label{eq:sk-bound-proof}
\end{eqnarray}
This proves \eqref{eq:sk-bound}.
\end{proof}

\begin{rem}
With the convention
$$
s_1(\zeta)
=
\int_{\mathbb R}\frac{d\mu(t)}{t-\zeta},
\qquad
\Impart\zeta>0,
$$
the upper half-plane boundary value satisfies
$$
\frac{1}{\pi}\Impart s_1(\xi+i0)=\rho(\xi)
$$
at points where $\mu$ has density $\rho$.
This is the sign convention used in the Stieltjes inversion formula below.
\end{rem}

\begin{rem}
The constant $C_\delta$ is independent of $k$.
This point is important for summing the random walk series.
\end{rem}

\section{Holomorphic continuation and finite-order expansion of the averaged diagonal resolvent}

We now use Lemma~\ref{lem:single-site-continuation} to continue the whole random walk series.

\begin{lem}
\label{lem:Mn-bound}
Fix $0<\delta<\delta_0$, and let $C_\delta$ be the constant from Lemma~\ref{lem:single-site-continuation}.
Put
$$
K_\delta
=
\max\{1,C_\delta\}(\delta_0-\delta)^{-1}.
$$
Then, for every $n\ge0$ and every $\zeta\in\Omega_\delta(I)$,
\begin{equation}
\label{eq:Mn-bound}
|M_n(\zeta)|
\le
|\Gamma_n(0,0)|K_\delta^{n+1}
\le
(2\sqrt q)^nK_\delta^{n+1}.
\end{equation}
\end{lem}

\begin{proof}
Let $\gamma\in\Gamma_n(0,0)$.
By Lemma~\ref{lem:single-site-continuation}, for every $x$ with $\nu_\gamma(x)>0$ we have
$$
|s_{\nu_\gamma(x)}(\zeta)|
\le
C_\delta(\delta_0-\delta)^{-\nu_\gamma(x)}.
$$
Hence
\begin{eqnarray}
\left|
\prod_{x:\,\nu_\gamma(x)>0}s_{\nu_\gamma(x)}(\zeta)
\right|
&\le&
C_\delta^{r_\gamma}
(\delta_0-\delta)^{-\sum_x\nu_\gamma(x)}
\nonumber\\
&\le&
K_\delta^{n+1},
\label{eq:walk-monomial-bound}
\end{eqnarray}
where
$$
r_\gamma=\#\{x\in\mathbb T_q:\nu_\gamma(x)>0\}
$$
is the number of visited vertices.
Here, we used $r_\gamma\le n+1$ and $\sum_x\nu_\gamma(x)=n+1$.
Summing over $\gamma\in\Gamma_n(0,0)$ gives the first inequality in \eqref{eq:Mn-bound}.
The second inequality follows from the spectral-radius bound \eqref{eq:closed-walk-spectral-radius}.
\end{proof}

We can now prove the main resolvent theorem.

\begin{proof}
We prove Theorem~\ref{thm:main-resolvent}.
Fix $0<\delta<\delta_0$.
Let
$$
Q_\delta=2\sqrt q\,K_\delta.
$$
Choose $\lambda_0$ so large that
$$
\lambda_0>2Q_\delta
\qquad
\text{and}
\qquad
\lambda_0>\frac{4\sqrt q}{\delta}.
$$

For $\lambda\ge\lambda_0$, Lemma~\ref{lem:Mn-bound} implies that the series
$$
\sum_{n=0}^{\infty}\lambda^{-n-1}M_n(\zeta)
$$
converges uniformly on $\Omega_\delta(I)$.
Indeed,
\begin{eqnarray}
\sum_{n=0}^{\infty}\lambda^{-n-1}|M_n(\zeta)|
&\le&
\sum_{n=0}^{\infty}\lambda^{-n-1}(2\sqrt q)^nK_\delta^{n+1}
\nonumber\\
&=&
\frac{K_\delta}{\lambda}
\sum_{n=0}^{\infty}
\left(\frac{Q_\delta}{\lambda}\right)^n,
\label{eq:uniform-series-bound}
\end{eqnarray}
and the geometric series converges because $Q_\delta/\lambda\le1/2$.
Therefore
\begin{equation}
\label{eq:F-lambda-def}
F_\lambda(\zeta)
=
\sum_{n=0}^{\infty}\lambda^{-n-1}M_n(\zeta)
\end{equation}
defines a holomorphic function on $\Omega_\delta(I)$.

We next compare $F_\lambda$ with the original averaged diagonal resolvent.
If $\Impart \zeta>2\sqrt q/\lambda$, then $\Impart (\lambda\zeta)>2\sqrt q$, and the Neumann series argument in Section~\ref{sec:randomwalkexpansion} gives the identity \eqref{eq:scaled-series}.
Hence, we have
$$
F_\lambda(\zeta)=m_\lambda(\lambda\zeta)
$$
for $\Impart \zeta>2\sqrt q/\lambda$.
Since $\lambda>2\sqrt q/\delta$, the set
$$
\Omega_\delta(I)\cap\{\Impart \zeta>2\sqrt q/\lambda\}
$$
is a nonempty open subset of $\Omega_\delta(I)\cap\mathbb C_+$.
On $\Omega_\delta(I)\cap\mathbb C_+$, both $F_\lambda(\zeta)$ and $m_\lambda(\lambda\zeta)$ are holomorphic.
Since $\Omega_\delta(I)\cap\mathbb C_+$ is connected, the identity theorem shows that
$$
F_\lambda(\zeta)=m_\lambda(\lambda\zeta)
$$
for every $\zeta\in\Omega_\delta(I)\cap\mathbb C_+$.
Thus, $F_\lambda$ is a holomorphic continuation of $\zeta\mapsto m_\lambda(\lambda\zeta)$ to $\mathbb C_+\cup\Omega_\delta(I)$.

We now derive the finite-order expansion.
Let $N\ge0$ and write
\begin{equation}
\label{eq:Rtilde-def}
\widetilde R_{N,\lambda}(\zeta)
=
\sum_{n=N+1}^{\infty}\lambda^{-n-1}M_n(\zeta).
\end{equation}
Then, by Lemma~\ref{lem:Mn-bound},
\begin{eqnarray}
|\widetilde R_{N,\lambda}(\zeta)|
&\le&
\frac{K_\delta}{\lambda}
\sum_{n=N+1}^{\infty}
\left(\frac{Q_\delta}{\lambda}\right)^n
\nonumber\\
&\le&
\frac{2K_\delta}{\lambda}
\left(\frac{Q_\delta}{\lambda}\right)^{N+1}
\nonumber\\
&=&
2K_\delta Q_\delta^{N+1}\lambda^{-N-2}.
\label{eq:Rtilde-bound}
\end{eqnarray}
Hence
\begin{equation}
\label{eq:scaled-main-expansion}
m_\lambda(\lambda\zeta)=\sum_{n=0}^{N}\lambda^{-n-1}M_n(\zeta)+\widetilde R_{N,\lambda}(\zeta)
\end{equation}
for $\zeta\in\Omega_\delta(I)$, with a uniform remainder bound of order $\lambda^{-N-2}$.
Replacing $\zeta$ by $z/\lambda$ gives \eqref{eq:main-resolvent-expansion} and \eqref{eq:main-resolvent-remainder}.
The formula \eqref{eq:Mn-formula-statement} is the same as \eqref{eq:Mn-def}, where $s_k$ is now understood as the continued function from Lemma~\ref{lem:single-site-continuation}.
The remaining assertions, namely $M_0=s_1$ and the vanishing of the odd coefficients, follow from Lemmas~\ref{lem:low-coefficients} and \ref{lem:odd-vanish}, respectively.
This completes the proof.
\end{proof}

\begin{rem}
The proof does not use a pointwise estimate of $|\lambda\omega_x-z|^{-1}$ away from the support in the neighborhood of $\lambda I$.
The only place where the stronger assumption is needed is the continuation of the single-site transforms in Lemma~\ref{lem:single-site-continuation}.
Once this lemma is available, the random walk series can be summed again because the Bethe lattice has only exponentially many closed walks of a given length.
\end{rem}

\section{Density of states}\label{sec:dos}

Recall the measure $N_\lambda$ in \eqref{eq:N-formal-def}.
It is characterized by the Stieltjes transform relation
\begin{equation}
\label{eq:dos-stieltjes}
m_\lambda(z)=\int_{\mathbb R}\frac{dN_\lambda(E)}{E-z},
\qquad\Impart z>0,
\end{equation}
see, for example, \cite{KirschMetzger2007}.
We use the convention \eqref{eq:dos-stieltjes} throughout this section.
With this convention, if $dN_\lambda$ is absolutely continuous on an interval and the upper half plane boundary value exists there, then its density is given by
$$
\frac{1}{\pi}\Impart m_\lambda(E+i0).
$$

\begin{proof}
We prove Theorem~\ref{thm:main-dos}.
By Theorem~\ref{thm:main-resolvent}, for every sufficiently large $\lambda$ the function $\zeta\mapsto m_\lambda(\lambda\zeta)$ is holomorphic on $\Omega_\delta(I)$.
Hence, the function
$$
E\mapsto m_\lambda(E)
$$
is holomorphic on the neighborhood $\lambda\Omega_\delta(I)$ of $\lambda I$.
Thus, its upper half-plane boundary value on $\lambda I$ exists and is equal to the restriction of this holomorphic continuation.

By the Stieltjes inversion formula with the sign convention in \eqref{eq:dos-stieltjes}, for every $\varphi\in C_0^\infty(\lambda I)$ one has
$$
\int_{\lambda I}\varphi(E)\,dN_\lambda(E)
=
\lim_{\varepsilon\downarrow0}\frac{1}{\pi}\int_{\lambda I}
\varphi(E)\Impart m_\lambda(E+i\varepsilon)\,dE .
$$
Since the upper half-plane branch of $m_\lambda$ extends holomorphically across $\lambda I$, this convergence is locally uniform on $\lambda I$.
Hence, we have
$$
\int_{\lambda I}\varphi(E)\,dN_\lambda(E)
=
\int_{\lambda I}
\varphi(E)\frac{1}{\pi}\Impart m_\lambda(E)dE.
$$
Therefore, $dN_\lambda$ is absolutely continuous on $\lambda I$, and its density is given by
\begin{equation}
\label{eq:density-boundary}
n_\lambda(E)=\frac{1}{\pi}\Impart m_\lambda(E),\qquad E\in\lambda I,
\end{equation}
where $m_\lambda(E)$ denotes the value at $E\in\lambda I$ of the holomorphic continuation.
Since $m_\lambda(E)$ is real analytic on $\lambda I$, the same is true for $n_\lambda(E)$.

The resolvent series in Theorem~\ref{thm:main-resolvent} converges uniformly on $\Omega_\delta(I)$.
Taking its imaginary part on $I$ therefore gives the convergent density expansion \eqref{eq:main-dos-convergent}.

Now let $\xi\in I$.
Substituting $z=\lambda\xi$ into \eqref{eq:main-resolvent-expansion} gives
\begin{equation}
\label{eq:boundary-at-xi}
m_\lambda(\lambda\xi)=\sum_{n=0}^{N}\lambda^{-n-1}M_n(\xi)
+R_{N,\lambda}(\lambda\xi).
\end{equation}
Here $M_n(\xi)$ denotes the value at $\xi\in I$ of the holomorphic continuation obtained in Theorem~\ref{thm:main-resolvent}.
Taking the imaginary part and dividing by $\pi$, we obtain \eqref{eq:main-dos-expansion} with $a_n$ defined by \eqref{eq:an-def} and with
$$
r_{N,\lambda}(\xi)=\frac{1}{\pi}\Impart R_{N,\lambda}(\lambda\xi).
$$
The bound \eqref{eq:rn-bound} follows from \eqref{eq:main-resolvent-remainder}.

It remains to identify $a_0$.
Since $M_0=s_1$, we have
$$
a_0(\xi)=\frac{1}{\pi}\Impart s_1(\xi),
$$
where $s_1(\xi)$ denotes the value at $\xi\in I$ of the holomorphic continuation of the single-site Stieltjes transform.
Since $s_1$ is the Stieltjes transform of $\mu$ and $\mu$ has density $\rho$ on $I$, the Stieltjes inversion formula gives
$$
\frac{1}{\pi}\Impart s_1(\xi)=\rho(\xi)
$$
for $\xi\in I$.
This proves \eqref{eq:a0-rho}.
The vanishing of odd $a_n$ follows from the vanishing of odd $M_n$ in Theorem~\ref{thm:main-resolvent}.
This completes the proof.
\end{proof}

\begin{rem}
The leading coefficient is immediate from the rescaled operator $V_\omega+\varepsilon A$: at $\varepsilon=0$ its root-averaged measure is $\mu$.
The nontrivial information in the expansion is in the higher coefficients and the uniform control of their series.
\end{rem}

\section{Girth and the comparison of local densities}\label{sec:girth}

We prove Theorem~\ref{thm:girth-comparison}.
The local coefficients are meaningful on any graph of bounded degree.
For the present comparison, let $G$ be a connected, simple, $D$-regular graph and let $o$ be the root.
The averaged diagonal resolvent at $o$ is obtained from the operator $A_G+\lambda V_\omega$, with the same independent single-site distribution $\mu$ as on the tree.
For $\zeta\in\mathbb C_+$, its strong-disorder coefficients are
\begin{equation}\label{eq:general-graph-walk-expansion}
M_n^{(G,o)}(\zeta)=(-1)^n
\sum_{\gamma\in\Gamma_n^{G}(o,o)}
\prod_{x:\,\nu_\gamma(x)>0}s_{\nu_\gamma(x)}(\zeta).
\end{equation}
All the single-site transforms in this formula are the same on $G$ and on $\mathbb T_{D-1}$.

The proof of the resolvent expansion in Section~\ref{sec:randomwalkexpansion} and its continuation in the following sections apply to $G$ as well.
Indeed, the original argument uses only the degree bound
\[
|\Gamma_n^G(o,o)|\le D^n,
\qquad \|A_G\|\le D,
\]
together with the estimates for $s_k$ in Lemma~\ref{lem:single-site-continuation}.
Thus, for the same $K_\delta$ as in Lemma~\ref{lem:Mn-bound},
\begin{equation}\label{eq:general-graph-Mn-bound}
|M_n^{(G,o)}(\zeta)|\le D^n K_\delta^{n+1}
\qquad (\zeta\in\Omega_\delta(I)).
\end{equation}
When $\lambda>DK_\delta$ and $\lambda>D/\delta$, the same identity-theorem argument gives a holomorphic continuation of the averaged resolvent and a uniformly convergent density expansion
\begin{equation}\label{eq:general-graph-density-series}
f_\lambda^{(G,o)}(\xi)
=\sum_{n=0}^{\infty}\lambda^{-n}a_n^{(G,o)}(\xi),
\qquad a_n^{(G,o)}(\xi)=\frac{1}{\pi}
\Impart M_n^{(G,o)}(\xi).
\end{equation}
The same bounds hold for the infinite $D$-regular tree, using $D^n$ in place of the sharper bound from \eqref{eq:closed-walk-spectral-radius}.
This common bound is useful because it is independent of the particular graph.

\begin{proof}[Proof of Theorem~\ref{thm:girth-comparison}]
Let $p:\mathbb T_{D-1}\to G$ be the universal covering map, with the reference vertex $\widetilde o$ mapped to $o$.
A walk in $G$ starting at $o$ has a unique lift starting at $\widetilde o$.
Every rooted closed walk of length $n<g$ becomes a rooted closed walk on the covering tree.
To see this, remove successive immediate backtracks from the walk.
If a nonempty reduced closed walk remained, it would contain a simple cycle of length at most $n$, contrary to the definition of $g$.
Thus, the lift returns to $\widetilde o$.

For $n<g$, this correspondence also preserves the occupation numbers.
In fact, if two occurrences of a vertex on the original walk had different lifts, the segment between them would give a nontrivial reduced closed walk of length less than $g$.
Conversely, the projection of a rooted closed walk on the tree is a rooted closed walk on $G$ with the same occupation numbers.
This proves \eqref{eq:girth-coeff-match} by \eqref{eq:general-graph-walk-expansion}.

At length $n=g$, the same correspondence holds for closed walks whose lifts return to $\widetilde o$.
Indeed, two occurrences of a vertex with different lifts would delimit a nontrivial reduced closed subwalk of length at least $g$.
At total length $g$, this could only be the whole walk, whose lifted endpoints coincide.
Thus the occupation numbers are still preserved.
The remaining walks are exactly the oriented simple cycles of length $g$ through $o$.
Indeed, after immediate backtracks are removed, a nontrivial closed walk of length $g$ must be a simple cycle of that length, and no backtracks can have been removed.
Each such cycle has occupation number $2$ at its initial and final vertex $o$ and occupation number $1$ at its other $g-1$ vertices.
Its contribution is $(-1)^g s_2 s_1^{g-1}$.
There are $c_g(o)$ such walks, which gives \eqref{eq:girth-first-difference}.
In particular, the sign $(-1)^g$ is needed when $g$ is odd.

Put $r=DK_\delta/\lambda$.
The assumptions $\lambda\ge\lambda_*$ give $r\le 1/2$ and also the nonempty intersection needed for the analytic continuation on $G$.
By \eqref{eq:general-graph-Mn-bound}, the equality of coefficients below order $g$, and the corresponding bound on the tree,
\begin{align*}
\left|f_\lambda^{(G,o)}(\xi)-
 f_\lambda^{(\mathbb T_{D-1})}(\xi)\right|
&\le\frac{2K_\delta}{\pi}
\sum_{n=g}^{\infty}r^n\\
&\le \frac{4K_\delta}{\pi}r^g.
\end{align*}
This proves \eqref{eq:girth-uniform-bound}.
Using \eqref{eq:girth-first-difference}, we may separate the term of order $g$ from this difference.
The remaining series starts at $n=g+1$.
Applying the same geometric bound with that lower limit proves \eqref{eq:girth-leading-dos} and \eqref{eq:girth-next-remainder}.
Since all estimates use only $D$ and the single-site law, the constants do not depend on $G$, $o$, or $g$.
\end{proof}

\begin{rem}\label{rem:girth-interpretation}
The estimate \eqref{eq:girth-uniform-bound} is exponentially small in the girth when $DK_\delta/\lambda<1$.
It concerns the \emph{averaged local density} at a fixed root, not normalized eigenvalue counts of balls in the infinite tree.
For a finite regular graph, averaging the disorder-averaged local spectral measures over all vertices gives the expected normalized eigenvalue counting measure:
\[
\frac{1}{|V(G)|}\sum_{o\in V(G)}N_\lambda^{(G,o)}(B)
=\frac{1}{|V(G)|}\mathbb E\operatorname{Tr}
 {\bf1}_B(A_G+\lambda V_\omega).
\]
A single-root measure need not equal this average on a non-transitive graph; it does equal it on a vertex-transitive graph.

The term in \eqref{eq:girth-leading-dos} is the first \emph{possible} difference: it vanishes when no shortest cycle passes through $o$, and it may also vanish at particular energies.
For odd $g$, it also shows how an odd-order coefficient can appear on a nonbipartite regular graph even though all such coefficients vanish on the tree.
For $g=4$, \eqref{eq:girth-first-difference} reduces to the four-cycle term in Remark~\ref{rem:four-cycles}.
The universal-cover relation used in the proof is a standard combinatorial observation.
What is specific to this theorem is the coefficient formula for the averaged local resolvent and the energy-uniform, pointwise comparison of locally analytic densities in the strong-disorder regime.
This is distinct from the local integrated-spectral-measure estimates, and their applications to growing random regular graphs, in \cite{Geisinger2015}.
\end{rem}

\section{The uniform distribution}
\label{sec:uniform-distribution}

\subsection{Explicit density coefficients through fourth order}

We consider the uniform distribution on $[-a,a]$, where $a>0$:
\begin{equation}
\label{eq:uniform-law}
d\mu(t)=\frac{1}{2a}{\bf1}_{[-a,a]}(t)\,dt.
\end{equation}
Fix an interval $I=(b_1,b_2)$ with $-a<b_1<b_2<a$.
Choose $\delta_0>0$ so that $I^\sharp=(b_1-\delta_0,b_2+\delta_0)\subset(-a,a)$.
On $I^\sharp$ the density is the constant $\rho=1/(2a)$, so the assumptions of Theorems~\ref{thm:main-resolvent} and \ref{thm:main-dos} hold.

The Stieltjes transform is
\begin{equation}
\label{eq:s1-uniform}
s_1(\zeta)=\frac{1}{2a}
\log\frac{a-\zeta}{-a-\zeta},\qquad\Impart\zeta>0,
\end{equation}
with the logarithm continued from the upper half-plane.
For $k\geq2$, we have
\begin{equation}
\label{eq:sk-uniform}
s_k(\zeta)
=\frac{1}{(k-1)!}\frac{d^{k-1}}{d\zeta^{k-1}}s_1(\zeta)
=\frac{1}{2a(k-1)}
\left\{\frac{1}{(-a-\zeta)^{k-1}}
-\frac{1}{(a-\zeta)^{k-1}}\right\}.
\end{equation}
In particular,
\begin{equation}
\label{eq:s2s3-uniform}
s_2(\zeta)=-\frac{1}{a^2-\zeta^2},\qquad
s_3(\zeta)=-\frac{\zeta}{(a^2-\zeta^2)^2}.
\end{equation}
For $-a<\xi<a$, let
\[
L(\xi)=\log\frac{a-\xi}{a+\xi},\qquad
D(\xi)=a^2-\xi^2.
\]
The upper boundary value is
\begin{equation}
\label{eq:s1-boundary-uniform}
s_1(\xi+i0)=\frac{L(\xi)+i\pi}{2a}.
\end{equation}
Thus $a_0(\xi)=1/(2a)$.
By \eqref{eq:M2-formula} and \eqref{eq:s2s3-uniform},
\begin{equation}
\label{eq:a2-uniform}
a_2(\xi)=\frac{1}{\pi}\Impart
\bigl((q+1)s_2(\xi)s_1(\xi+i0)\bigr)
=-\frac{q+1}{2aD(\xi)}.
\end{equation}

We next obtain an explicit fourth-order coefficient.

\begin{thm}\label{thm:uniform-fourth}
For the uniform distribution \eqref{eq:uniform-law} and $\xi\in I\Subset(-a,a)$, the fourth-order coefficient in Theorem~\ref{thm:main-dos} is
\begin{equation}
\label{eq:a4-uniform}
a_4(\xi)=\frac{q(q+1)}{2aD(\xi)^2}
\left\{1-\frac{\xi}{a}L(\xi)\right\}.
\end{equation}
Consequently, as $\lambda\to\infty$,
\begin{equation}
\label{eq:uniform-dos-fourth}
n_\lambda(\lambda\xi)
=\frac{1}{2a\lambda}
-\frac{q+1}{2aD(\xi)\lambda^3}
+\frac{q(q+1)}{2aD(\xi)^2\lambda^5}
\left\{1-\frac{\xi}{a}L(\xi)\right\}
+O(\lambda^{-7}).
\end{equation}
The remainder is uniform for $\xi$ in any fixed compact subinterval of $(-a,a)$.
\end{thm}

\begin{proof}
From Lemma~\ref{lem:low-coefficients},
\[
M_4=(q+1)s_3s_2
+q(q+1)s_3s_1^2
+q(q+1)s_2^2s_1.
\]
Here $s_2(\xi)$ and $s_3(\xi)$ are real for $\xi\in I$.
Therefore the first term makes no contribution to $\Impart M_4(\xi)$.
From \eqref{eq:s1-boundary-uniform},
\[
\frac{1}{\pi}\Impart s_1(\xi+i0)=\frac{1}{2a},\qquad
\frac{1}{\pi}\Impart s_1(\xi+i0)^2=\frac{L(\xi)}{2a^2}.
\]
Hence
\[
\begin{aligned}
a_4(\xi)
&=q(q+1)\left\{
\frac{s_3(\xi)L(\xi)}{2a^2}
+\frac{s_2(\xi)^2}{2a}\right\}\\
&=\frac{q(q+1)}{2aD(\xi)^2}
\left\{1-\frac{\xi}{a}L(\xi)\right\}.
\end{aligned}
\]
This proves \eqref{eq:a4-uniform}.
All odd coefficients vanish.
Applying Theorem~\ref{thm:main-dos} with $N=5$ gives \eqref{eq:uniform-dos-fourth}, including the stated remainder.
\end{proof}

\begin{rem}\label{rem:uniform-geometry}
The factor $q(q+1)$ comes from two types of four-step walks: $0\to x\to0\to y\to0$ with $x\ne y$, and $0\to x\to u\to x\to0$ with $u\ne0$.
Each type has $q(q+1)$ choices.
The remaining four-step walks repeat one edge; their contribution to $M_4$ is real for the uniform distribution and does not appear in $a_4$.
On a graph with four-cycles there is an additional term $c_4(o)s_2s_1^3$ in $M_4$ as noted in Remark~\ref{rem:four-cycles}.
Its contribution to the fourth-order density coefficient for \eqref{eq:uniform-law} is
\begin{equation}
\label{eq:a4-fourcycle-extra}
-\frac{c_4(o)}{8a^3D(\xi)}
\left\{3L(\xi)^2-\pi^2\right\}.
\end{equation}
This term is absent on the Bethe lattice.
Thus the fourth-order coefficient shows a concrete difference in local graph geometry, not only a dependence on the vertex degree.
\end{rem}

\begin{rem}
The formulas are local in energy and are not uniform as $\xi$ approaches $\pm a$.
The density of the uniform law is not analytic at these endpoints.
No statement about the edge behavior is made here.
\end{rem}

%%%%%%%%%%%%%%%

\section*{Statements and Declarations}

\subsection*{Competing Interests}
The author declares that he has no competing interests.

\subsection*{Data availability}
No datasets were generated or analyzed during the current study.

\subsection*{Funding}
The author received no specific funding for this work.

\subsection*{Author Contributions}
The author contributed to all parts of the manuscript.

%\subsection*{Acknowledgements}
%The author thanks the anonymous referees for helpful comments.

\end{document}